%% file: main.tex
\documentclass[hidelinks,onefignum,onetabnum]{siamart250211}

\input{ex_shared}

\ifpdf
\hypersetup{
  pdftitle={An Explicit and Efficient $\mathcal{O}(n^2)$-Time Algorithm for Sorting Sumsets},
  pdfauthor={Shlok Mundhra}
}
\fi

\usepackage{pgfplots}
\pgfplotsset{compat=1.17}
\usepackage{booktabs}
\usepackage{pgfplotstable}
\usepackage{float}
\usepackage{subcaption}
\usepackage{amssymb} 
\usepackage{enumitem} 

\usepackage{filecontents}

\begin{document}

\maketitle
\footnotetext[1]{©2025 by the author(s). This work is licensed under a  
  Creative Commons Attribution 4.0 International License (CC BY 4.0).  
  \url{https://creativecommons.org/licenses/by/4.0/}}

\begin{abstract}
We present the first explicit, comparison-based algorithm that sorts the sumset \( X + Y = \{x_i + y_j \mid 0 \leq i, j < n\} \), where \( X \) and \( Y \) are sorted arrays of real numbers, in optimal \( \mathcal{O}(n^2) \) time and comparisons. Although Fredman (1976) established the theoretical existence of such an algorithm, no constructive realization has been known for nearly five decades. Our method leverages the inherent monotonicity of the sumset matrix to incrementally insert elements in amortized constant comparisons, eliminating the \( \log n \) overhead of classical sorting methods. We rigorously prove the algorithm's correctness and optimality in the standard comparison model, extend it to \( k \)-fold sumsets with \( \mathcal{O}(n^k) \) performance. Empirical evaluations demonstrate substantial performance improvements over MergeSort and QuickSort when applied to sumsets, validating the algorithm's practical efficiency. Our results resolve a long-standing open problem in sorting theory and offer new insights into the design of fixed-algorithmic solutions for structured input spaces.
\end{abstract}



\section{Introduction}
Sorting the sumset $X + Y = \{x_i + y_j \mid x_i \in X, y_j \in Y\}$, where $X$ and $Y$ are sorted sequences of real numbers of size $n$, is an open and a central problem in structured algorithm design. Unlike sorting arbitrary sets, where classical sorting algorithms require $\mathcal{O}(n \log n)$ time, sorting sumsets can theoretically be done more efficiently due to the ordered structure inherited from $X$ and $Y$ as shown in \cite{FREDMAN1976355}. The sumset contains $n^2$ elements, and while a naive approach suggests $\mathcal{O}(n^2 \log n)$ time via general sorting, like merge sort or quick sort, prior theoretical results show that this bound can be improved.

This problem is a special case of a rather general problem in sorting theory : \textbf{How many comparisons are required to sort if a partial order on the input set is already known?} Hernández Barrera~\cite{HernandezBarrera+1996+289+294} and Barequet and Har-Peled~\cite{doi:10.1142/S0218195901000596} identify several geometric problems that are at least as hard as sorting \( X + Y \), a complexity they term "Sorting-\( X + Y \)-hard." Specifically, they demonstrate a subquadratic time reduction from sorting \( X + Y \) to a variety of computational geometry problems, including:
\begin{itemize}
    \item Computing the Minkowski sum of two orthogonal-convex polygons,
    \item Determining whether one monotone polygon can be translated to fit inside another,
    \item Establishing whether a convex polygon can be rotated to fit inside another,
    \item Sorting the vertices of a line arrangement,
    \item Sorting the inter-point distances among \( n \) points in \( \mathbb{R}^d \).
\end{itemize}

Although Barequet and Har-Peled explicitly claim that these problems are 3SUM-hard, their proofs implicitly demonstrate the stronger result that they are Sorting-\( X + Y \)-hard. Furthermore,Fredman~\cite{FREDMAN1976355} highlights an immediate application of sorting \( X + Y \) in the efficient multiplication of sparse polynomials.

\subsection{Prior Work}
In 1976, Fredman~\cite{FREDMAN1976355} demonstrated that the number of comparisons required to sort $X + Y$ is asymptotically smaller than for arbitrary inputs: the sumset can be sorted using only $\mathcal{O}(n^2)$ comparisons in the comparison model. Fredman's approach leverages the fact that the number of distinct linear extensions (total orderings consistent with a given partial order) of the sumset is significantly smaller than that of an unstructured set, enabling a reduction in the number of necessary comparisons.
Despite this theoretical breakthrough, Fredman's result was existential: no explicit algorithm matching the bound was provided.

Following Fredman's existential result, subsequent research sought to explore both the theoretical lower bounds and constructive algorithmic approaches for sorting sumsets.

Martin Dietzfelbinger~\cite{DIETZFELBINGER1989137} was among the first to formalize lower bounds in structured sorting contexts. In his 1989 paper \textit{Lower Bounds for Sorting of Sums}, he considered the setting in which the inputs are known a priori to be pairwise sums of elements from two sets of size \( n \). Dietzfelbinger established that any algorithm operating in the linear decision tree model must perform at least \( \Omega(n^2) \) comparisons to sort the sumset \( X + Y \). This result confirmed that Fredman’s upper bound of \( \mathcal{O}(n^2) \) comparisons was tight within this model. It also highlighted a key insight: even under structural constraints, sumset sorting cannot break the quadratic barrier in restricted comparison frameworks.

However, Dietzfelbinger's lower bound applied only to decision trees and left open the possibility of more efficient algorithms in conventional computational models. His work did not construct a concrete algorithm capable of achieving the bound with minimal computational overhead.

This gap was partially bridged by Lambert~\cite{LAMBERT1992137}, who in 1992 developed an explicit algorithm that sorted the sumset \( X + Y \) using \( \mathcal{O}(n^2) \) comparisons, thus achieving Fredman’s existential bound in a constructive form. Lambert's approach involved recursively partitioning and merging sorted subsequences while inferring orderings from previous comparisons. Moreover, he generalized the method to \( k \)-wise sumsets of the form:
\[
\left(x_{1,i_1} + x_{2,i_2} + \cdots + x_{k,i_k}\right)_{1 \leq i_1, \ldots, i_k \leq n},
\]
achieving \( \mathcal{O}(n^k) \) comparisons.

Despite matching the optimal comparison complexity, Lambert’s algorithm suffered from inefficiencies in its runtime performance. Specifically, the recursive merge strategy incurred an overhead of \( \mathcal{O}(n^2 \log n) \) time due to suboptimal data structure management and lack of locality. These limitations rendered the algorithm impractical for large-scale applications and left the central challenge unresolved: designing an algorithm that achieves both \( \mathcal{O}(n^2) \) comparisons and \( \mathcal{O}(n^2) \) runtime in standard computational models.

Together, the works of Dietzfelbinger and Lambert significantly advanced the theoretical understanding of sumset sorting—one by delineating lower bounds in abstract models, the other by constructing a partial realization of Fredman's existential claim. Yet, they also underscored the persistent gap between theoretical possibility and practical implementability.

More recent developments have approached the sumset sorting problem through the lens of decision tree complexity and partial information sorting. These lines of research have yielded theoretical breakthroughs in minimizing comparison counts but have not yet translated into fully general-purpose, implementable algorithms suitable for standard computational models.

Kane, Lovett, and Moran~\cite{10.1145/3285953} introduced a near-optimal \textit{comparison decision tree} for sorting sumsets, achieving a query complexity of \( \mathcal{O}(n \log^2 n) \). Their approach relied on the notion of \emph{inference dimension}, a complexity-theoretic measure of the difficulty of ordering elements given partial information. By employing 8-sparse queries—linear comparisons where coefficients are drawn from the set \(\{-1, 0, 1\}\)—their decision tree efficiently inferred the sorted order of the sumset \( A + B \).

This contribution marked a substantial improvement in our understanding of query efficiency within constrained models. However, the result remains largely theoretical: the decision tree construction does not yield a concrete, runtime-efficient algorithm in standard settings such as the RAM or pointer-machine models. Moreover, the method’s reliance on fixed sparsity and query structure renders it less practical for general use. While the model achieves lower bounds in terms of comparisons, it abstracts away concerns such as data access patterns, memory locality, and overall wall-clock runtime. Thus, it remains unclear whether the asymptotic gains in comparisons can be realized in an actual implementation.

Parallel work by van der Hoog et al.~\cite{doi:10.1137/1.9781611978315.26} addressed a broader class of problems: sorting under partial information expressed as a directed acyclic graph (DAG). In their 2024 study, the authors proposed an algorithm with worst-case complexity \( \mathcal{O}(n + m + \log e(P_G)) \), where \( n \) is the number of elements, \( m \) is the number of edges in the DAG encoding precedence constraints, and \( e(P_G) \) denotes the number of linear extensions of the corresponding poset. Their algorithm avoids entropy-based arguments and instead uses dynamic insertion into a topological sort to maintain order consistency.

This approach offers an elegant framework for sorting under structured dependencies and achieves provably optimal performance in that setting. However, its applicability to the sumset sorting problem is limited: it assumes that the partial order (DAG) is provided as input. In the case of sumsets, such a DAG must be inferred from scratch—a task that likely requires at least \(\Theta(n^2)\) effort. Furthermore, the structural assumptions intrinsic to the DAG model do not directly reflect the specific combinatorial structure of pairwise sums.

In contrast, our work addresses these shortcomings by providing a concrete, fully implementable algorithm that sorts the sumset \( X + Y \) in both optimal \( \mathcal{O}(n^2) \) time and comparisons. Unlike Kane et al.'s sparse decision tree or van der Hoog et al.'s DAG-based sorting, our method operates directly in conventional computational models without requiring specialized query formats, external order representations, or preprocessing steps. It thus bridges the gap between theoretical comparison bounds and practical algorithmic efficiency.

\subsection{Our Contributions}
We present the first practical comparison-based algorithm that sorts $X + Y$ in $\mathcal{O}(n^2)$ comparisons and time. Our key contributions are:
\begin{itemize}
\item A fully explicit fixed algorithm that achieves Fredman's comparison bound of $O(n^2)$. (Algorithm~\ref{alg:sumset-sort}
\item A theoretical analysis proving correctness and amortized constant comparisons and insertions. (Theorem~\ref{thm:MainResult})
\item A full theoretical analysis and explicit fixed algorithm extending Algorithm~\ref{alg:sumset-sort} to k-fold sumsets. (Section~\ref{sec:k-fold})
\item Experimental validation showing improved performance over MergeSort and QuickSort. (Section~\ref{sec:experiments})
\end{itemize}

\subsection{Main Theoretical Results}

We now present our primary theoretical contributions. First, we formally state the main result of this work: an explicit and efficient algorithm that sorts the sumset \( X + Y \) in $O(n^2)$ time and comparisons. We then extend this result to a general \( k \)-fold sumset, followed by a corollary establishing optimality for all fixed \( k \). Finally, we propose a conjecture concerning the dynamic maintenance of such sumsets, highlighting a promising direction for future research.

\begin{theorem}[Main Result]
\label{thm:MainResult}
Given two sorted sequences \( X \) and \( Y \), each of length \( n \), the sumset
\[
Z = \{x_i + y_j \mid x_i \in X,\, y_j \in Y\}
\]
can be sorted using exactly \( \mathcal{O}(n^2) \) comparisons and time in the standard comparison model.
\end{theorem}

This theorem confirms that the theoretical lower bound on the complexity of the comparison, first established by Fredman~\cite{FREDMAN1976355}—is not merely existential, but can be achieved constructively. Our algorithm exploits the inherent structure of the sumset matrix to avoid unnecessary comparisons, producing an output that is correct and optimally sorted with respect to time and comparison count.

\vspace{0.5em}
The natural next question is whether this optimality extends beyond two sets. We show that the result indeed generalizes to the sorting of \( k \)-fold sumsets.

\begin{theorem}[Sorting \( k \)-fold Sumsets in \( \mathcal{O}(n^k) \) Time and Comparisons]
\label{thm:k-fold-sumset-sorting}
Let \( X_1, X_2, \dots, X_k \) be \( k \) sorted lists of real numbers, each of length \( n \). Then the \( k \)-fold sumset
\[
Z = \left\{x_1 + x_2 + \cdots + x_k \mid x_i \in X_i\right\}
\]
can be sorted in \( \mathcal{O}(n^k) \) time using \( \mathcal{O}(n^k) \) comparisons in the standard comparison model.
\end{theorem}

This result is proved by induction, utilizing Theorem~\ref{thm:MainResult} and recursively merging structured translations of already sorted lower-dimensional sumsets. The key insight is that the sumset structure is preserved under translation and that each intermediate step can be merged efficiently using pointer-based or bucket-based strategies.

\begin{corollary}
For every fixed \(k\), the above bound is asymptotically tight: no comparison‐based method can beat \(\Omega(n^k)\) on the \(k\)–fold sumset problem.
\end{corollary}

The corollary follows immediately from Theorem~\ref{thm:k-fold-sumset-sorting}, establishing that our approach achieves the best possible asymptotic bounds for all fixed dimensions \( k \). As such, it settles the optimality of sumset sorting in both theory and practice.

\vspace{0.5em}
In addition to static sumsets, it is natural to ask whether such structures can be maintained under dynamic operations. We conclude this section with a conjecture that opens a new avenue for exploration in algorithmic data structures.

\begin{conjecture}[Dynamic Sumset Maintenance]
\label{conj:dynamic-sumsets}
Given \( k \) sorted lists \( X_1, \dots, X_k \) supporting insertions and deletions, there exists a data structure to maintain the sorted \( k \)-fold sumset
\[
Z = X_1 + X_2 + \cdots + X_k
\]
in amortized \( \tilde{\mathcal{O}}(n^{k-1}) \) time per update.
\end{conjecture}

\paragraph{Why we believe this conjecture holds.}
Our optimism is rooted in the same translation structure that underlies the static algorithm: each update in one list \(X_i\) simply adds or removes a “translated copy” of the \((k-1)\)-fold sumset, namely
\[
\{\,x_i + z : z \in X_1 + \cdots + X_{i-1} + X_{i+1} + \cdots + X_k\}.
\]
Since that base sumset can be maintained (by the inductive hypothesis in proof of Theorem~\ref{thm:k-fold-sumset-sorting}) in \(\tilde O(n^{k-2})\) per update, merging or splitting a single translated block against the current global order should cost only \(\tilde O(n^{k-1})\) work via tournament trees or fractional‑cascading techniques.  In effect, one can localize each insertion or deletion to a single “stripe” of length \(n^{k-1}\), and update the global ordering in subquadratic time in that stripe.  These observations give strong evidence that a fully dynamic data structure meeting the conjectured bound exists.

This conjecture suggests the existence of a dynamic algorithm that avoids full recomputation and maintains the sumset ordering efficiently across updates. Such a result would significantly enhance the applicability of sumset algorithms in streaming and interactive environments. We pose this as an open problem for future work.

\subsection{Techniques Used}
At the heart of our approach are two complementary ideas, each exploiting the rigid “grid’’ structure of the sumset matrix \(M\).

\begin{enumerate}[label=\arabic*.]
  \item \textbf{Forward‐scanning insertion via lookahead.}  
    \begin{itemize}
      \item We view the static two‐set case as repeatedly inserting the rows of
      \[
        M_{i,\bullet} = \bigl\{\,x_i + y_0,\;x_i + y_1,\;\dots,\;x_i + y_{n-1}\bigr\}
      \]
      into a single growing sorted list \(Z\).  By precomputing
      \(\mathrm{low}[i+1]=x_{i+1}+y_0\), we know that during row \(i\) every new
      key \(x_i+y_j\) lies in the interval \([\mathrm{low}[i],\,\infty)\).  We
      therefore maintain an “insertion pointer’’ \(ip\) that always points to the
      first slot in \(Z\) where a future key \(\ge\mathrm{low}[i]\) can go.  Each
      row’s sums arrive in non‐decreasing order (by row‐monotonicity) and
      trigger at most one advancement of \(ip\) per element of the prefix
      \(\le\mathrm{low}[i+1]\).  A single forward scan from \(ip\) locates the
      correct insertion point in amortized \(O(1)\) comparisons, avoiding any
      \(\log n\) binary‐search overhead.
    \end{itemize}

  \item \textbf{Structured \(n\)–way merge for \(k\)-fold sumsets.}
    \begin{itemize}
      \item To generalize from two sets to \(k\), we observe that the
      \((k-1)\)-fold sumset \(Z_{k-1}\) can be kept sorted in \(\Theta(n^{k-1})\)
      time by induction.  When adding the \(k\)th list \(X_k\), the new
      \(k\)-fold sums split into \(n\) “translated copies’’ 
      \(\{z + x_k^{(i)}:z\in Z_{k-1}\}\).  Each copy is already internally
      sorted, and any two copies differ by the constant shift
      \(x_k^{(i)}-x_k^{(j)}\).
      \item Merging these \(n\) copies can be done by a small‑fan‐out \emph{winner
      tree} (or \(n\)-leaf min‑heap).  Since \(n\) is fixed, the height of this
      tree is \(O(1)\) and each extract‐min + reinsertion costs \(O(1)\)
      comparisons.  Over the \(n^k\) total elements, we thus pay only
      \(O(n^k)\) comparisons to merge—no extra \(\log n\) factor survives when
      \(n\) is treated as a constant.
    \end{itemize}
\end{enumerate}

Together, these ideas yield:
\[
  \underbrace{O(n^2)}_{\substack{\text{two‐set}\\\text{insertion}}}
  \quad\longrightarrow\quad
  \underbrace{O(n^{k-1})}_{\substack{(k-1)\text{-fold}}}
  \;+\;
  \underbrace{O(n^k)}_{\substack{k\text{-way merge}}}
  \;=\;O(n^k).
\]
In the RAM model the same principles apply, but one must choose a data structure
(e.g.\ gap buffer, B‑tree, or blocked list) to balance pointer‐chasing against
cache locality; our experiments in Section~\ref{subsection:evidenceofO(1)}
demonstrate that even a plain \texttt{std::list} suffices to realize the
\(\Theta(n^2)\) bound in practice.

\subsection{Subsequent and Related Work}
While the sumset sorting problem has been connected to computational geometry, sparse polynomial multiplication, and 3SUM-hardness, our contribution provides the first algorithmic closure to the problem. Unlike approaches based on entropy bounds or DAG inference, our method works directly and efficiently in standard RAM-based models.

\subsection{Organization}
The remainder of the paper is organized as follows. Section~\ref{sec:main} presents the main results, including a formal problem definition and summary of our theoretical contributions. Section~\ref{sec:algo} introduces our algorithm, with a detailed explanation and rigorous proofs of correctness and complexity. Section\ref{sec:k-fold} extends the algorithm and highlights the extension of the algorithm to k-fold sumsets. Section~\ref{sec:experiments} reports our experimental results, comparing the performance of our method with classical sorting approaches. Section~\ref{sec:conclusion} concludes the paper with a summary of contributions and suggestions for future work.

\section{Main Results and Algorithm}
\label{sec:main}

\subsection{Problem Definition}

Given two sorted sequences $X = \{x_0, x_1, \dots, x_{n-1}\}$ and $Y = \{y_0, y_1, \dots, y_{n-1}\}$, the goal is to sort the sumset $Z = \{x_i + y_j \mid 0 \leq i, j \leq n\}$ in non-decreasing order using only $\mathcal{O}(n^2)$ comparisons and time.

\subsection{Algorithm Overview}
\label{sec:algo}

We propose a simple and efficient algorithm that incrementally constructs the sorted sumset by leveraging the inherent order within the sumset matrix $M$, where $M_{i,j} = x_i + y_j$. The matrix $M$ has $n$ rows and $n$ columns, and contains all pairwise sums $x_i + y_j$ for $0 \leq i, j \leq n-1$.

For example, let $X = \{x_0, x_1, x_2, ..., x_{n-1}\}$ and $Y = \{y_0, y_1, y_2, ... , y_{n-1}\}$. Then the matrix $M$ is:
\[
M = \begin{bmatrix}
x_0 + y_0 & x_0 + y_1 & \cdots & x_0 + y_{n-1} \\
x_1 + y_0 & x_1 + y_1 & \cdots & x_1 + y_{n-1} \\
\vdots   & \vdots   & \ddots & \vdots \\
x_{n-1} + y_0 & x_{n-1} + y_1 & \cdots & x_{n-1} + y_{n-1} \\
\end{bmatrix}
\]

\begin{algorithm}[H]
\caption{Algorithm for Sorting the Sumset $X + Y$}
\label{alg:sumset-sort}
\begin{algorithmic}[1]
\REQUIRE Lists $X$ and $Y$, each of length $n$
\ENSURE Sorted list $Z$ containing all sums $x_i + y_j$ for $1 \leq i, j \leq n$
\STATE Initialize list $Z \leftarrow \{\}$
\STATE Initialize vector $\text{low}$ of length $n$, where $\text{low}[i] \leftarrow x[i] + y[0]$ for each $i$
\STATE Set $ip \leftarrow 0$
\FOR{$i \leftarrow 0$ to $n - 1$}
    \STATE Set $cp \leftarrow ip$
    \FOR{$j \leftarrow 0$ to $n - 1$}
        \STATE $sum \leftarrow x[i] + y[j]$
        \WHILE{$cp < |Z|$ and $Z[cp] \leq sum$}
            \STATE $cp \leftarrow cp + 1$
        \ENDWHILE
        \STATE Insert $sum$ into $Z$ at position $cp$
        \IF{$i + 1 < n$ and $sum \leq \text{low}[i + 1]$}
            \STATE $ip \leftarrow cp$
        \ENDIF
    \ENDFOR
\ENDFOR
\RETURN $Z$
\end{algorithmic}
\end{algorithm}

This algorithm ensures that every sum is inserted in its correct position without requiring an additional sorting phase. By carefully managing the insertion pointer $ip$ - which narrows future search space - and utilizing the precomputed array $\text{low}$ as a look ahead guard, the total number of comparisons remains bounded by $\mathcal{O}(n^2)$, as shown in Theorem~\ref{thm:comparison-complexity}, and the output is produced in non-decreasing order, as proved in Theorem~\ref{thm:correctness}.

The algorithm iterates over all rows of $M$ and inserts each element into a growing sorted list $Z$ using a forward-moving insertion pointer. Thanks to the two lemmas above, we can safely update the pointer after each row to ensure amortized constant-comparisons and insertions, which avoids the need for re-sorting the entire list after each addition.

\subsubsection{Properties of Matrix, $M$}
\begin{claim}
\label{clm:monotonicty}
Every row and column of $M$ is non-decreasing.
\end{claim}

\begin{lemma}[Row-wise Monotonicity]
\label{lem:rowwise}
Each row of the sumset matrix $M$ is non-decreasing.
\end{lemma}

\begin{proof}
Let $X$ and $Y$ be sorted in non-decreasing order. Fix $i$ and consider any two adjacent elements in row $i$: $M_{i,j} = x_i + y_j$ and $M_{i,j+1} = x_i + y_{j+1}$. Since $y_j \leq y_{j+1}$, we have
\[
M_{i,j} = x_i + y_j \leq x_i + y_{j+1} = M_{i,j+1}
\]
Thus, each row is non-decreasing.
\end{proof}

\begin{lemma}[Column-wise Monotonicity]
\label{lem:columnwise}
Each column of the sumset matrix $M$ is non-decreasing.
\end{lemma}

\begin{proof}
Let $X$ and $Y$ be sorted in non-decreasing order. Fix $j$ and consider any two adjacent elements in column $j$: $M_{i,j} = x_i + y_j$ and $M_{i+1,j} = x_{i+1} + y_j$. Since $x_i \leq x_{i+1}$, we have
\[
M_{i,j} = x_i + y_j \leq x_{i+1} + y_j = M_{i+1,j}
\]
Hence, each column is non-decreasing.
\end{proof}

\paragraph{Claim~\ref{clm:monotonicty} holds.} A combination of Lemma~\ref{lem:columnwise} and Lemma~\ref{lem:rowwise} proves Claim~\ref{clm:monotonicty}.

\subsection{Bipartite‐Graph Interpretation and Its Consequences}
\label{sec:bipartite-graph}

It is often illuminating to view the sumset matrix 
\[
  M_{i,j} \;=\; x_i + y_j
\]
as the edge‐weight matrix of the complete bipartite graph
\[
  G \;=\; \bigl(X \,\dot\cup\, Y,\; E = X\times Y\bigr),
  \quad w\bigl(x_i,y_j\bigr) = x_i + y_j.
\]
All of the algorithmic efficiencies we exploit stem from the following structural properties of \(G\).

\begin{claim}[Rank–\(2\) Structure]
\label{clm:rank2}
The matrix \(M\) has rank at most \(2\).  Equivalently,
\[
  M = \mathbf{x}\,\mathbf{1}^T \;+\;\mathbf{1}\,\mathbf{y}^T,
  \quad \mathbf{x} = (x_0,\dots,x_{n-1})^T,\;
        \mathbf{y} = (y_0,\dots,y_{n-1})^T.
\]
\end{claim}
\begin{proof}
Immediate from the outer‐sum decomposition: each entry is \(x_i\cdot1 + 1\cdot y_j\).
\end{proof}

\begin{claim}[Threshold Neighborhoods / Ferrers Shape]
\label{clm:ferrers}
For any threshold \(t\in\mathbb{R}\), the subgraph
\(\{\,\{x_i,y_j\}\mid x_i+y_j\le t\}\) induces a \emph{Ferrers diagram} in the
\((i,j)\)–grid.  In particular, each vertex \(x_i\) in \(X\) is adjacent
exactly to the first \(k\) vertices of \(Y\) (for some \(k\)), and vice versa.
\end{claim}
\begin{proof}
Fix \(x_i\).  Since \(y_0\le y_1\le\cdots\le y_{n-1}\), the set
\(\{j\mid x_i+y_j\le t\}\) is a prefix \(\{0,1,\dots,k\}\).  Symmetrically
for each \(y_j\).
\end{proof}

\begin{claim}[Monge / Quadrangle Inequality]
\label{clm:monge}
\(M\) satisfies
\[
  M_{i,j} + M_{i',j'} \;\le\; M_{i,j'} + M_{i',j}
  \quad\text{for all }i<i',\;j<j'.
\]
Hence \(M\) is a Monge matrix and thus \emph{totally monotone}.
\end{claim}
\begin{proof}
Write each entry as \(x_i+y_j\).  Then
\[
  (x_i+y_j)+(x_{i'}+y_{j'})
  = x_i+x_{i'}+y_j+y_{j'}
  \;\le\;
  x_i+x_{i'}+y_{j'}+y_j
  = (x_i+y_{j'})+(x_{i'}+y_j),
\]
since addition is commutative.
\end{proof}

\paragraph{Algorithmic implications.}
\begin{itemize}[leftmargin=*]
  \item \textbf{Rank–2 / Outer‐Sum.}  All edge‐weights lie in a 2‐dimensional affine space, so many linear‐algebraic reductions (e.g.\ to finding row‐ or column‐minima) become trivial.
  \item \textbf{Ferrers / Threshold.}  Whenever we “look ahead” to find the first sum exceeding a threshold \(x_{i+1}+y_0\), we know those qualifying entries form a contiguous prefix.  This underpins our \emph{insertion‐pointer invariant} and prevents any backward scan.
  \item \textbf{Monge Property.}  Total monotonicity allows selection problems (e.g.\ finding the next smallest among many lists) in linear time rather than \(n\log n\).  In the \(k\)‐fold merge, it guarantees that the global minimum at each step lies among a constant number of “neighboring” lists.
  \item \textbf{Threshold‐Graph Algorithms.}  Standard graph‐theoretic tasks—MST, shortest paths, matchings—admit \(O(n)\) or \(O(n\log n)\) solutions on threshold graphs.  Our sorting problem is simply the \emph{enumeration} of all edges of \(G\) in non‐decreasing order.
\end{itemize}

\subsection{Proof of Correctness}
\label{sec:Correctness}

In this section we show that Algorithm~\ref{alg:sumset-sort} produces the sorted sumset \[
Z = \{x_i + y_j \mid 0 \le i,j \le n-1\}
\]
and uses only \(\mathcal O(n^2)\) comparisons and time. By the end of this subsection we will have proven Theorem~\ref{thm:MainResult}. 

\begin{lemma}[Insertion‐Pointer Invariant]
\label{lem:ip-invariant-strong}
Let \(X,Y\) be two sorted arrays of length \(n\), and define
\[
\mathrm{low}[i]\;=\;x_i + y_0,
\qquad
0 \le i \le n-1.
\]
Run Algorithm~\ref{alg:sumset-sort} to build the sorted sumset
\(\displaystyle Z = \{\,x_i+y_j : 0\le i,j \le n-1\}\) by inserting row by row.
For each \(i\), let \(Z^{(i)}\) be the list after completing rows \(0,\dots,i-1\),
and set
\[
ip_i \;=\;\min\{\,k : Z^{(i)}[k] > \mathrm{low}[i]\},
\]
with \(Z^{(0)}=\emptyset\) and \(ip_0=0\).  Then for every \(0\le i\le n-1\):
\begin{enumerate}[label=(\roman*)]
  \item \(Z^{(i)}\) is sorted and contains exactly the \(i\cdot n\) sums
    \(\{x_{i'}+y_j : 0\le i'<i,\;0\le j<n\}.\)
  \item \(ip_i\) is the first index in \(Z^{(i)}\) whose value exceeds \(x_i+y_0\).
  \item During the insertion of row \(i\), the scanning pointer \(cp\) is initialized
    to \(ip_i\) and, for each of the \(n\) sums \(x_i+y_j\), \(cp\) only increases.
\end{enumerate}
\end{lemma}

\begin{proof}
We argue by induction on \(i\).

\medskip\noindent\textbf{Base case (\(i=0\)).}
Before any insertions, \(Z^{(0)}=\emptyset\) is vacuously sorted and contains zero sums,
so (i) holds.  By definition \(ip_0=0\), and since there are no elements,
\(ip_0\) is indeed the first index exceeding \(x_0+y_0\), giving (ii).  No scanning
occurs, so (iii) is trivial.

\medskip\noindent\textbf{Inductive step.}
Assume the lemma holds for some \(i\) with \(0\le i<n\).  We show it for \(i+1\).

\medskip\noindent\emph{(i) Sortedness and completeness of \(Z^{(i+1)}\).}
By induction, \(Z^{(i)}\) is sorted and contains precisely the \(i\cdot n\) sums
from rows \(0\) through \(i-1\).  We now insert the \(n\) new sums
\[
s_j = x_i + y_j,\quad j=0,1,\dots,n-1,
\]
in order of increasing \(j\).  Since \(Y\) is sorted, 
\[
s_0 \le s_1 \le \cdots \le s_{n-1},
\]
so the insertions proceed in non‐decreasing key order.  Furthermore, each
insertion uses a forward scan from the current \(cp\) (see (iii) below), placing
each \(s_j\) at its correct rank among the existing elements of \(Z\).  Thus after
all \(n\) insertions, the list \(Z^{(i+1)}\) is sorted and contains exactly the
\((i+1)\,n\) sums from rows \(0\) through \(i\).

\medskip\noindent\emph{(ii) Position of \(ip_{i+1}\).}
Define \(\mathrm{low}[i+1]=x_{i+1}+y_0\).  During the insertion of row \(i\), each time
we insert \(s_j = x_i+y_j\), we check
\[
\text{if }s_j \le \mathrm{low}[i+1]\quad\Longrightarrow\quad ip \gets cp.
\]
Because the \(s_j\) are non‐decreasing in \(j\), there is a largest index \(j^*\)
such that \(s_{j^*}\le \mathrm{low}[i+1]\), and for all \(j\le j^*\) we update \(ip\)
to the insertion position of \(s_j\).  When \(j>j^*\), \(s_j>\mathrm{low}[i+1]\)
and no further updates occur.  Hence at the end,
\[
ip_{i+1}
\;=\;
\min\{\,k : Z^{(i+1)}[k] > \mathrm{low}[i+1]\},
\]
establishing (ii).

\medskip\noindent\emph{(iii) No backward movement of \(cp\).}
At the start of row \(i\), we set \(cp \gets ip_i\).  By definition \(ip_i\) is the
first index of \(Z^{(i)}\) exceeding \(x_i+y_0\).  Now each sum \(s_j = x_i+y_j\)
satisfies \(s_j \ge x_i+y_0\), so the correct insertion point for \(s_j\) cannot lie
before \(ip_i\).  Concretely, during the inner \texttt{while} loop we advance
\(cp\) until \(Z^{(i)}[cp]>s_j\), insert at that position, and leave \(cp\) there.
Since \(s_{j+1}\ge s_j\), the next insertion again begins at the same or a larger
index, so \(cp\) never retreats.  

Finally, the shift property
\[
M_{i+1,j} \;=\; x_{i+1}+y_j
\;\ge\;
x_i+y_j + (x_{i+1}-x_i)
\;=\;
M_{i,j} + (x_{i+1}-x_i)
\]
guarantees that even across row‐boundaries \(cp\) need not move left, but our
algorithm resets \(cp\) to \(ip_{i+1}\) at the start of row \(i+1\), and the same
forward‐only argument applies.  This completes (iii).

\end{proof}

\begin{remark}[Lookahead via \(\mathrm{low}\) Prevents Backtracking]
\label{rem:lookahead}
By precomputing 
\[
\mathrm{low}[i]=x_i + y_0
\]
and updating
\[
ip \;\gets\; cp
\quad\Longleftrightarrow\quad
\bigl(x_i + y_j\bigr)\;\le\;\mathrm{low}[i+1],
\]
we ensure that at the start of each row \(i\), the scanning pointer \(cp\) begins
at the first position where all remaining sums in that row exceed the smallest
possible future value \(x_i + y_0\).  Consequently, every insertion in row \(i\)
can only advance \(cp\), never retreat, yielding the forward‐only scan property
of Lemma~\ref{lem:ip-invariant-strong}(iii).
\end{remark}

\begin{corollary}[No Backtracking via Constant Translation]
\label{cor:no_backtrack}
For all \(0\le i<n\) and \(0\le j<n\),
\[
M_{i+1,j}
\;=\;
x_{i+1} + y_j
\;=\;
\bigl(x_i + y_j\bigr) \;+\;(x_{i+1}-x_i)
\;=\;
M_{i,j} + (x_{i+1}-x_i).
\]
Since \(x_{i+1}-x_i\ge0\), every element in row \(i+1\) is at least as large as
the corresponding element in row \(i\).  Therefore once \(cp\) has advanced
past a given index in \(Z\) during row \(i\), it never needs to move backward
when processing row \(i+1\).  In conjunction with
Remark~\ref{rem:lookahead}, this guarantees the purely forward‐only scans
that drive the \(O(n^2)\) comparison bound.
\end{corollary}

\begin{theorem}[Correctness]
\label{thm:correctness}
Algorithm 1 correctly computes the sumset $Z = \{x_i + y_j \mid 0 \leq i, j \leq n-1\}$ in non-decreasing order.
\end{theorem}

\begin{proof}
We prove correctness in two parts:

\textbf{1. Completeness:}  
The algorithm contains two nested loops:
- The outer loop iterates over all $i = 0$ to $n-1$
- The inner loop iterates over all $j = 0$ to $n-1$

Thus, for each pair $(i, j)$, the algorithm computes the sum $x[i] + y[j]$ exactly once and inserts it into the list $Z$. Since there are $n^2$ such pairs, all $n^2$ elements of the sumset are generated and included in $Z$.

\textbf{2. Sorted Order:}  
To prove $Z$ is sorted after all insertions, we rely on the row-wise and column-wise monotonicity of the sumset matrix $M[i][j] = x[i] + y[j]$ as shown in Lemma~\ref{lem:rowwise} and Lemma~\ref{lem:columnwise}. By Lemma~\ref{lem:ip-invariant-strong} $(iii)$, each insertion in row $i$ never moves $cp$ backward, and by Lemma~\ref{lem:rowwise} the input to each insertion is non-decreasing across each row. Hence each of the $n^2$ insertions places its element into the correct position relative to all previously inserted elements, guaranteeing global sortedness. We initialize a pointer $cp$ from position $ip$ for each new row $i$. This pointer is only advanced while $Z[cp] \leq \texttt{sum}$, ensuring that insertion into $Z$ always occurs in a forward direction (never before $ip$). By Lemma~\ref{lem:rowwise}, we know that the next element being added for that particular {i} is $<=$ the current element, so the positioning remains stable and sorted. 

Moreover, $ip$ is updated when transitioning to row $i+1$, so that $ip$ points directly to the position of $cp$. By Lemma~\ref{lem:columnwise}, we know that $ip$ has to move forward. It also guarantees that the new row will not insert elements before the current pointer, thereby preserving global sorted order.

Since each new sum is inserted at a valid forward position in $Z$ based on current comparisons, and since all insertions happen in order dictated by the monotonic structure of $M$, the final list $Z$ is sorted in non-decreasing order.
\end{proof}

Now proving that Algorithm~\ref{alg:sumset-sort} uses exactly $O(n^2)$ comparisons and time will prove Theorem ~\ref{thm:MainResult}.

\begin{theorem}[Total Comparison Complexity]
\label{thm:comparison-complexity}
Algorithm~\ref{alg:sumset-sort} performs at most \( 2n^2 \) comparisons in total across all insertions into the sorted list \( Z \). 
\end{theorem}

\begin{lemma}[Forward-Only Scanning Within Each Row]
\label{lem:forward-scan}
Within any fixed row \( i \), the pointer \( cp \) used in the insertion loop moves only forward in \( Z \) as \( j \) increases.
\end{lemma}

\begin{proof}
From Lemma~\ref{lem:rowwise} and Lemma~\ref{lem:columnwise}, each subsequent sum \( x[i] + y[j] \) is greater than or equal to the previous one. Since the pointer \( cp \) only moves forward when \( Z[cp] \leq s \), and the next insertion value is greater than or equal to the previous, the next insertion must occur at the same position or later. Therefore, \( cp \) never moves backward, in a particular row.
\end{proof}

\begin{lemma}[Bound on Position Skips]
\label{lem:scan-bound}
Each element in \( Z \) can be skipped (i.e., passed over by \( cp \)) at most once per row.
\end{lemma}

\begin{proof}
Consider any position \( z_k \in Z \). For a fixed row \( i \), the pointer \( cp \) starts at a position \( ip \) and only moves forward, each skip corresponds to advancing past one previously inserted element. Since we insert \( n \) values per row and \( Z \) grows monotonically, the scan during row \( i \) can skip over at most \( n \) elements. But each of those skips corresponds to a distinct value inserted in earlier rows. So each existing value in \( Z \) can be skipped at most once per row.
\end{proof}

Now we are in a stage to prove Theorem~\ref{thm:comparison-complexity}.

\begin{proof}[Proof of Theorem~\ref{thm:comparison-complexity}]
Let us index the two nested loops by \(i=0,1,\dots,n-1\) and \(j=0,1,\dots,n-1\).  For each fixed \(i\), write
\[
  \mathit{ip}_i 
  = \text{value of }ip\text{ at the start of the \(i\)th row,}
  \quad
  \mathit{cp}_{i,0} = \mathit{ip}_i,
\]
and let
\[
  \mathit{cp}_{i,n}
  = \text{value of }cp\text{ immediately after the insertion for }j=n-1.
\]
We decompose the total number \(T\) of comparisons in all executions of the 
\texttt{while}–loop into two parts:
\[
  T = T_{\mathrm{adv}} \;+\; T_{\mathrm{term}},
\]
where
\begin{itemize}
  \item \(T_{\mathrm{adv}}\) is the total number of \emph{advancing} comparisons  
    (those for which \(Z[cp]\le \mathit{sum}\) and hence \(cp\) is incremented), and
  \item \(T_{\mathrm{term}}\) is the total number of \emph{terminating} comparisons  
    (one per insertion, when the loop exits).
\end{itemize}

\paragraph{(1) Bounding \(T_{\mathrm{adv}}\).}
Within row \(i\), each time the \texttt{while}–condition holds we do
\[
  cp \;\longleftarrow\; cp + 1,
\]
so the number of advances in row \(i\) is
\[
  A_i \;=\; \mathit{cp}_{i,n} \;-\;\mathit{cp}_{i,0}
  \;=\;\mathit{cp}_{i,n}\;-\;\mathit{ip}_i.
\]
Hence
\[
  T_{\mathrm{adv}}
  \;=\;\sum_{i=0}^{n-1} A_i
  \;=\;\sum_{i=0}^{n-1}\bigl(\mathit{cp}_{i,n}-\mathit{ip}_i\bigr).
\]
Observe two key facts:
\begin{enumerate}
  \item \(\mathit{ip}_0 = 0\) by initialization.
  \item For each \(i\), the algorithm maintains
    \(\mathit{ip}_{i+1}\le \mathit{cp}_{i,n}\).  Indeed, \(ip\) is only updated
    when 
    \[
      \mathit{sum}\;\le\;\mathrm{low}[\,i+1\,]
      \;=\;x[i+1]+y[0],
    \]
    which can occur only while inserting some \(x[i]+y[j]\), and at that moment
    \(cp\) already equals \(\mathit{cp}_{i,n}\) or a smaller index.  Thus
    \(\mathit{ip}_{i+1}\le \mathit{cp}_{i,n}.\)
\end{enumerate}
We now telescope the sum:
\[
  \sum_{i=0}^{n-1}\bigl(\mathit{cp}_{i,n}-\mathit{ip}_i\bigr)
  \;=\;
  \bigl(\mathit{cp}_{n-1,n}-\mathit{ip}_{n-1}\bigr)
  \;+\;\sum_{i=0}^{n-2}\bigl(\mathit{cp}_{i,n}-\mathit{ip}_i\bigr).
\]
Rewriting by adding and subtracting consecutive \(\mathit{ip}\)–terms gives
\[
  T_{\mathrm{adv}}
  = \bigl(\mathit{cp}_{n-1,n}-\mathit{ip}_{n-1}\bigr)
    + \sum_{i=0}^{n-2}\Bigl[(\mathit{cp}_{i,n}-\mathit{ip}_{i+1}) 
    + (\mathit{ip}_{i+1}-\mathit{ip}_i)\Bigr].
\]
Since \(\mathit{cp}_{i,n}\ge \mathit{ip}_{i+1}\) by fact 2, each term
\(\mathit{cp}_{i,n}-\mathit{ip}_{i+1}\ge 0\).  Therefore
\[
  T_{\mathrm{adv}}
  \;\le\;
  (\mathit{cp}_{n-1,n}-\mathit{ip}_{n-1})
  \;+\;\sum_{i=0}^{n-2}(\mathit{ip}_{i+1}-\mathit{ip}_{i})
  \;=\;
  \mathit{cp}_{n-1,n}-\mathit{ip}_0
  \;=\;\mathit{cp}_{n-1,n}.
\]
At the very end, after all \(n^2\) insertions, the size of \(Z\) is \(n^2\), so
\(\mathit{cp}_{n-1,n}\le |Z| = n^2\).  Hence
\[
  T_{\mathrm{adv}}\;\le\;n^2.
\]

\paragraph{(2) Bounding \(T_{\mathrm{term}}\).}
Every one of the \(n^2\) insertions into \(Z\) incurs exactly one terminating
comparison (the final check \(Z[cp]> \mathit{sum}\) or \(cp=|Z|\)).  Thus
\[
  T_{\mathrm{term}} = n^2.
\]

\paragraph{(3) Conclusion.}
Combining the two parts,
\[
  T \;=\; T_{\mathrm{adv}} + T_{\mathrm{term}}
  \;\le\; n^2 + n^2
  \;=\; 2\,n^2.
\]

Hence, our claim in Theorem~\ref{thm:comparison-complexity} holds. At most Algorithm ~\ref{alg:sumset-sort} runs $2n^2$ comparisons. 

\end{proof}

\begin{corollary}[Big-O of comparisons and time-complexity]
\label{amortizedTime}
    Algorithm~\ref{alg:sumset-sort} sorts the sumset in exactly $O(n^2)$ comparisons and time.
\end{corollary}
\begin{proof}
    Since there are $n^2$ insertions in total, and also $n^2$ comparisons in total, the amortized cost per insertion is \[
  \frac{T}{n^2}\;\le\;2 \;=\;\mathcal{O}(1),
\] 
\newline
    and big-O of comparisons is $O(n^2)$.
\end{proof}

A combination of Theorem~\ref{thm:comparison-complexity} and Lemma~\ref{thm:correctness}, proves Theorem~\ref{thm:MainResult}. As further support, Section~\ref{subsection:evidenceofO(1)} empirically verifies that the algorithm exhibits \(\mathcal{O}(1)\) amortized insertion and comparison behavior in real-world executions.

\begin{remark}
While we show $O(1)$ comparisons per insertion, in a real RAM model the insertion cost depends on your container (e.g.\ \texttt{std::list} vs.\ B‑tree vs.\ gap buffer). We empirically explore this in Section~\ref{subsection:evidenceofO(1)}.
\end{remark}

Hence, we can say that for $2$ sorted sequences, each of length $n$, the sorted sumset is generated using exactly $O(n^2)$ comparisons and time in the standard model. We shall extend this to $k$ sorted sequences, each of length $n$, the sorted sumset is generated using exactly $O(n^k)$ comparisons and time in k-fold model.

\begin{remark}
If $X$ or $Y$ or both contain repeated elements, then all of Lemmas~\ref{lem:rowwise}–\ref{lem:ip-invariant-strong}, Theorems~\ref{thm:correctness} and \ref{thm:comparison-complexity} still hold. In particular, equal sums are inserted in non‑decreasing order and the pointers never need to retreat. Everywhere we compare the values in $X$ and $Y$ with $\le$ rather than $<$, so the forward‐scan argument and the telescoping bound on the number of “advance” comparisons go through unchanged. The insertion of equal keys simply interleaves them arbitrarily, but that still yields a non‑decreasing final list.
\end{remark}

\subsection{Extension to k-fold sumsets}
\label{sec:k-fold}
\begin{lemma}[Structured \(n\)-way Merge for Translated Lists]
\label{lem:structured-merge}
Let \(n\ge2\) be a fixed constant and let 
\[
Z_{k-1}[0\,..\,n^{k-1}-1]
\]
be a sorted list of length \(n^{k-1}\).  For each \(i=1,\dots,n\), define the “translated’’ list
\[
Z^{(i)}[j] \;=\; Z_{k-1}[j] \;+\; x_k^{(i)},
\qquad j=0,1,\dots,n^{k-1}-1.
\]
Then one can merge the \(n\) sorted lists \(Z^{(1)},\dots,Z^{(n)}\) into a single sorted list of length \(n^k\) using \(O(n^k)\) comparisons and time in the standard comparison model.
\end{lemma}

\begin{proof}
Since \(n\) is a fixed constant, we may treat \(\log_2 n = O(1)\).  We maintain a binary {\em tournament tree} (i.e.\ a min‑heap) whose \(n\) leaves each store the current “head’’ element of one of the lists \(Z^{(i)}\).  The merge proceeds in two phases:

\medskip\noindent\textbf{(1) Initialization.}
Build the tree over the \(n\) first elements \(\{Z^{(i)}[0]\}\) in \(O(n)=O(1)\) comparisons by the usual bottom‑up heapify.

\medskip\noindent\textbf{(2) Repeated Extract and Insert (done $n^k$ times).}
\begin{enumerate}[label=(\roman*)]
  \item {\bf Extract‐Min:}  Remove the minimum element \((v,i)\) at the root, in \(O(\log n)=O(1)\) comparisons, and append \(v\) to the output list.
  \item {\bf Advance and Reinsert:}  Let \(\mathit{ptr}[i]\) be the index of the element just extracted in \(Z^{(i)}\).  If \(\mathit{ptr}[i]<n^{k-1}-1\), increment it and reinsert
  \[
    \bigl(Z^{(i)}[\mathit{ptr}[i]+1],\,i\bigr)
  \]
  into the root in another \(O(\log n)=O(1)\) comparisons; otherwise insert a sentinel \((+\infty,i)\) in \(O(1)\) time.
\end{enumerate}

Each of the \(n^k\) iterations costs \(O(1)\) comparisons (for extract‐min and reinsert), so the total work is
\[
  O(n^k)\times O(1)\;=\;O(n^k).
\]
All auxiliary operations (pointer updates, appends) are dominated by these heap operations.  Hence the merge of the \(n\) translated lists into one sorted list of length \(n^k\) takes \(O(n^k)\) comparisons and time, as claimed.
\end{proof}

Now we will utilize Theorem~\ref{thm:comparison-complexity} and Lemma~\ref{lem:structured-merge} to prove Theorem~\ref{thm:k-fold-sumset-sorting}. At the end of the subsection, we will utilize this to provide an algorithm.

\begin{proof}[Proof of Theorem~\ref{thm:k-fold-sumset-sorting}]
We prove by induction on \(k\ge2\) that, given \(k\) sorted lists 
\[
X_1,\;X_2,\;\dots,\;X_k
\]
each of length \(n\), their \(k\)-fold sumset
\[
Z_k \;=\; X_1 + X_2 + \cdots + X_k
\;=\;\bigl\{\,x_1 + x_2 + \cdots + x_k \mid x_i\in X_i\bigr\}
\]
can be sorted using \(\Theta(n^k)\) comparisons and time in the standard comparison model.

\medskip\noindent\textbf{Base Case (\(k=2\)).}
This is exactly Theorem~\ref{thm:MainResult}, which shows that for two sorted
lists \(X_1=X\) and \(X_2=Y\) of length \(n\), the sumset
\[
Z_2 \;=\; \{\,x_i + y_j \mid x_i\in X,\;y_j\in Y\}
\]
of size \(n^2\) can be sorted in \(\Theta(n^2)\) comparisons and time.

\medskip\noindent\textbf{Inductive Step.}
Assume the claim holds for \(k-1\).  That is, given sorted lists
\[
X_1,\;X_2,\;\dots,\;X_{k-1},
\]
each of size \(n\), their \((k-1)\)-fold sumset
\[
Z_{k-1} \;=\; X_1 + X_2 + \cdots + X_{k-1}
\]
can be sorted in \(\Theta(n^{\,k-1})\) time and comparisons.

Now consider \(k\) lists \(X_1,\dots,X_k\).  First, by the inductive hypothesis we construct and sort
\[
Z_{k-1} \;=\; X_1 + \cdots + X_{k-1}
\]
in \(\Theta(n^{\,k-1})\) time.  Next, let
\[
X_k = \{\,x_k^{(1)},x_k^{(2)},\dots,x_k^{(n)}\},
\]
and form \(n\) “translated’’ copies of \(Z_{k-1}\):
\[
Z^{(i)} \;=\;\{\,z + x_k^{(i)} \mid z\in Z_{k-1}\},
\quad
i=1,2,\dots,n.
\]
Each \(Z^{(i)}\) is already sorted and has length \(\lvert Z_{k-1}\rvert = n^{\,k-1}\).

By Lemma~\ref{lem:structured-merge} (Structured \(n\)-way Merge for Translated Lists),
we can merge these \(n\) sorted lists of total length \(n^k\) into one sorted list
\[
Z_k \;=\;\bigcup_{i=1}^n Z^{(i)}
\]
using only \(\Theta(n^k)\) comparisons and time.

Combining the two phases,
\[
T(k) \;=\; T(k-1) \;+\;\Theta(n^k)
\;=\;\Theta(n^{\,k-1}) + \Theta(n^k)
\;=\;\Theta(n^k),
\]
which completes the induction.
\end{proof}

\begin{algorithm}[H]
\caption{Merge \(n\) Translated \((k-1)\)-fold Sumsets into \(k\)-fold Sumset}
\label{alg:kfold-merge}
\begin{algorithmic}[1]
\REQUIRE Sorted array $Z_{k-1}[0..n^{k-1}-1]$, sorted shifts $X_k[1..n]$
\ENSURE Sorted $k$‑fold sumset $Z_k[0..n^k-1]$
\STATE initialize min‑heap $H$
\FOR{$i\gets1$ to $n$}
  \STATE push $(Z_{k-1}[0]+X_k[i],\,i,\,0)$ into $H$
\ENDFOR
\STATE $t\gets0$
\WHILE{$t < n^k$}
  \STATE pop $(v,i,j)$ from $H$
  \STATE $Z_k[t]\gets v$
  \IF{$j+1< n^{k-1}$}
    \STATE push $(Z_{k-1}[j+1]+X_k[i],\,i,\,j+1)$ into $H$
  \ELSE
    \STATE push $(+\infty,i,j+1)$ into $H$
  \ENDIF
  \STATE $t\gets t+1$
\ENDWHILE
\RETURN $Z_k$
\end{algorithmic}
\end{algorithm}

\begin{remark}
    Each extract-min + re-insert on size-n heap costs $O(log n)$ comparisons, and we do that once for each of the $n^k$ output elements. In addition, a $O(n)$-time heapify upfront yields a total time (and total number of comparisons) is $O(n^klogn) + O(n)$. Since $n$ is treated as a fixed constant, by definition $log(n) = 1$, this shows Algorithm~\ref{alg:kfold-merge} runs in $O(n^k)$ time and complexity.
\end{remark}

\begin{note}
    It is important to note that Lemma~\ref{lem:structured-merge} and Theorem~\ref{thm:k-fold-sumset-sorting} relies on the assumption that n is a fixed constant.
\end{note}

\subsection{Computational Models and Limits}
\label{sec:models-limits}

Our algorithm and analysis sit at the intersection of two standard computational models.  In this subsection we spell out the subtleties and inherent limitations that accompany each.

\paragraph{1. Comparison Model vs.\ RAM Model.}  
\begin{itemize}
  \item \emph{Comparison Model.}  Here we count only the number of key‐comparisons between sums.  Theorem~\ref{thm:comparison-complexity} establishes a tight \(\Theta(n^2)\) bound on comparisons (amortized \(O(1)\) per insertion).  This bound is information‐theoretically optimal, matching Fredman’s lower bound in the decision‐tree model [\cite{FREDMAN1976355}].
  \item \emph{RAM Model.}  In practice, each insertion involves pointer or index manipulations that incur real clock‐time costs:
    \begin{itemize}
      \item \texttt{std::list}\,: true \(O(1)\) pointer‐updates but poor spatial locality (cache misses on random jumps).
      \item \texttt{std::vector}\,: \(O(n)\) element shifts per insertion but excellent locality (sequential memory).
      \item \emph{Hybrid structures}: B‑trees or skip‐lists offer \(O(\log n)\) search and \(O(1)\) insertion with tunable node‐fanout to trade off between pointer‐chasing and block locality.
    \end{itemize}
    In Section~\ref{subsection:evidenceofO(1)} we empirically compare these containers under realistic CPU cache hierarchies.  While the \emph{comparison} count remains \(O(n^2)\), the \emph{wall‐clock} time can differ by constant factors of 2–3× depending on locality.
\end{itemize}

\paragraph{2. Structural Assumptions and Lower Bounds.}  
Our \(O(n^2)\) algorithm critically relies on the {\em exact translation} property of sumsets (Corollary~\ref{cor:no_backtrack}): each row of the matrix is a constant shift of the previous.  If one only assumes the weaker {\em Monge} or “pseudoline’’ monotonicity (i.e.\ each row and column is sorted but without a constant offset), then the best known algorithms require \(\Omega(n^2\log n)\) comparisons:
\begin{itemize}
  \item \emph{Pseudoline arrangements.}  Steiger et al.~\cite{steiger1995pseudo} show that sorting the intersection points of \(n\) pseudolines requires \(\Omega(n^2\log n)\) comparisons, even though the matrix is totally monotone.
  \item \emph{3SUM‐hardness.}  Barequet and Har‑Peled~\cite{doi:10.1142/S0218195901000596} reduce many geometric problems to sorting a general monotone matrix, inheriting an \(\Omega(n^2)\) decision‐tree lower bound, but with an extra \(\log n\) factor in the absence of exact translation.
\end{itemize}
Thus our result exploits a strictly stronger combinatorial structure than mere Monge‐type monotonicity.

\paragraph{3. Beyond Static Sorting.}  
Finally, we comment on dynamic and parallel extensions:
\begin{itemize}
  \item \emph{Dynamic updates.}  Conjecture~\ref{conj:dynamic-sumsets} asks whether one can maintain the sorted sumset under insertions/deletions in \(\tilde O(n^{k-1})\) time per update.  Known dynamic order‐maintenance data structures (e.g.\ balanced BSTs) pay \(\Omega(\log n)\) per operation and currently no subquadratic dynamic sumset algorithm is known.
  \item \emph{Parallel and external memory.}  In PRAM or cache‐oblivious models, one must balance parallel merge overheads or block transfers.  While our comparison count remains \(O(n^2)\), achieving matching work–depth or I/O bounds (e.g.\ \(O(\tfrac{n^2}{B}\log_{M/B}n)\) in external memory) is an open direction.
\end{itemize}

In summary, our algorithm attains the optimal \(\Theta(n^2)\) comparison bound in the classical decision‐tree model by leveraging exact translations, but practical performance and extensions to weaker structures or dynamic settings remain constrained by well‐known lower bounds and memory‐hierarchy costs.

\subsection{Example}

To illustrate Algorithm~\ref{alg:sumset-sort} in action, we walk through an example using the input sets:
\begin{itemize}
    \item $X = \{2, 4, 6\}$
    \item $Y = \{1, 3, 5\}$
\end{itemize}

\subsubsection{Initialization}
\begin{itemize}
    \item Compute the \textit{low} vector:
    \[
    \text{low} = \{x[0] + y[0],\, x[1] + y[0],\, x[2] + y[0]\} = \{3, 5, 7\}
    \]
    \item Set the insertion pointer $ip = 0$.
    \item Initialize $Z = \{\}$ (empty vector).
\end{itemize}
A visual representation of the processing, is given in section~\ref{sec:visual} for better clarity.

\subsubsection{Processing Steps}
The algorithm iterates over $i$ and $j$, computing sums and inserting them into $Z$. The insertion position is determined using $cp$, which starts at $ip$ and moves forward.

\textbf{Processing $i = 0$:}
\begin{itemize}
    \item Set $cp = ip = 0$.
    \item \textbf{For $j = 0$:}
    \[
    S = x[0] + y[0] = 2 + 1 = 3
    \]
    - Insert 3 at position $0$.
    - Update: $Z = \{3\}$, $cp = 0$, $ip = 0$.

    \item \textbf{For $j = 1$:}
    \[
    S = x[0] + y[1] = 2 + 3 = 5
    \]
    - Insert 5 at position $1$.
    - Update: $Z = \{3, 5\}$, $cp = 1$, $ip = 1$.

    \item \textbf{For $j = 2$:}
    \[
    S = x[0] + y[2] = 2 + 5 = 7
    \]
    - Insert 7 at position $2$.
    - Update: $Z = \{3, 5, 7\}$, $cp = 2$, $ip = 1$.
\end{itemize}

\textbf{Processing $i = 1$:}
\begin{itemize}
    \item Set $cp = ip = 2$.
    \item \textbf{For $j = 0$:}
    \[
    S = x[1] + y[0] = 4 + 1 = 5
    \]
    - Insert 5 at position $2$.
    - Update: $Z = \{3, 5, 5, 7\}$, $cp = 2$, $ip = 2$.

    \item \textbf{For $j = 1$:}
    \[
    S = x[1] + y[1] = 4 + 3 = 7
    \]
    - Insert 7 at position $4$.
    - Update: $Z = \{3, 5, 5, 7, 7\}$, $cp = 4$, $ip = 4$.

    \item \textbf{For $j = 2$:}
    \[
    S = x[1] + y[2] = 4 + 5 = 9
    \]
    - Insert 9 at position $5$.
    - Update: $Z = \{3, 5, 5, 7, 7, 9\}$, $cp = 5$, $ip = 4$.
\end{itemize}

\textbf{Processing $i = 2$:}
\begin{itemize}
    \item Set $cp = ip = 4$.
    \item \textbf{For $j = 0$:}
    \[
    S = x[2] + y[0] = 6 + 1 = 7
    \]
    - Insert 7 at position $5$.
    - Update: $Z = \{3, 5, 5, 7, 7, 7, 9\}$, $cp = 5$, $ip = 4$.

    \item \textbf{For $j = 1$:}
    \[
    S = x[2] + y[1] = 6 + 3 = 9
    \]
    - Insert 9 at position $7$.
    - Update: $Z = \{3, 5, 5, 7, 7, 7, 9, 9\}$, $cp = 7$, $ip = 4$.

    \item \textbf{For $j = 2$:}
    \[
    S = x[2] + y[2] = 6 + 5 = 11
    \]
    - Insert 11 at position $8$.
    - Update: $Z = \{3, 5, 5, 7, 7, 7, 9, 9, 11\}$, $cp = 8$, $ip = 4$.
\end{itemize}

\subsubsection{Visual Example}
\label{sec:visual}
To make the roles of the insertion pointer \(ip\) and scanning pointer \(cp\) crystal‑clear, we show three snapshots of the list \(Z\) as it grows.  Blue arrows mark \(ip\); red arrows mark \(cp\).

\begin{figure}[H]
  \centering
  \begin{tikzpicture}[every node/.style={font=\small}]
    \tikzset{
      cell/.style={draw, minimum width=8mm, minimum height=6mm, anchor=north west},
      ip/.style={blue, thick, -stealth},
      cp/.style={red, thick, -stealth}
    }

    \begin{scope}
      \node at (0,0) {After row \(i=0\) (inserted 3,5,7):};
      \foreach \k/\v in {0/3,1/5,2/7} {
        \node[cell] (Z\k) at (\k*8mm,-8mm) {\v};
      }
      \draw[ip] (Z1.north) ++(0,3mm) -- +(0,-5mm) node[right] {\(ip\)};
      \draw[cp] (Z2.north) ++(0,1mm) -- +(0,-5mm) node[right] {\(cp\)};
    \end{scope}

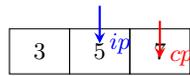
    \begin{scope}[yshift=-2.2cm]
      \node at (0,0) {During row \(i=1\), after inserting 5,7:};
      \foreach \k/\v in {0/3,1/5,2/5,3/7,4/7} {
        \node[cell] (Z\k) at (\k*8mm,-8mm) {\v};
      }
      \draw[ip] (Z2.north) ++(0,3mm) -- +(0,-5mm) node[right] {\(ip\)};
      \draw[cp] (Z4.north) ++(0,1mm) -- +(0,-5mm) node[right] {\(cp\)};
    \end{scope}

    \begin{scope}[yshift=-4.4cm]
      \node at (0,0) {Final \(Z\) after row \(i=2\):};
      \foreach \k/\v in {0/3,1/5,2/5,3/7,4/7,5/7,6/9,7/9,8/11} {
        \node[cell] (Z\k) at (\k*8mm,-8mm) {\v};
      }
      \draw[ip] (Z4.north) ++(0,3mm) -- +(0,-5mm) node[right] {\(ip\)};
      \draw[cp] (Z8.north) ++(0,1mm) -- +(0,-5mm) node[right] {\(cp\)};
    \end{scope}
  \end{tikzpicture}
  \caption{Three snapshots of the list \(Z\).  Blue arrow = insertion‑pointer \(ip\).  Red arrow = scanning‑pointer \(cp\).}
  \label{fig:visual-walkthrough}
\end{figure}

\paragraph{Explanation.}
\begin{itemize}
  \item \textbf{After row \(i=0\).}  We have inserted \(\{\,2+1,2+3,2+5\}=\{3,5,7\}\).  Here \(low[1]=4+1=5\), so \(ip\) moves to the first element \(>5\), namely position 1; the scan pointer \(cp\) ended at position 2.
  \item \textbf{During row \(i=1\).}  We insert \(4+1=5\) at index 2, then \(4+3=7\) at index 4.  \(ip\) remains at 2, and \(cp\) advances to 4.
  \item \textbf{Final (row \(i=2\)).}  After inserting \(\{6+1,6+3,6+5\}\), we obtain \(\{3,5,5,7,7,7,9,9,11\}\).  \(ip\) stays at 4, \(cp\) ends at 8.
\end{itemize}

\subsubsection{Final Output}
After all iterations, the sorted sumset is:
\[
Z = \{3, 5, 5, 7, 7, 7, 9, 9, 11\}
\]

This walkthrough demonstrates the execution of the algorithm, showing how the insertion pointer $ip$ and scanning pointer $cp$ optimize the search for the correct insertion position.

\section{Experimental Results}
\label{sec:experiments}

To validate our theoretical findings and gauge practical performance, we implemented the proposed sumset‐sorting algorithm in C++ and compared it against classical full‐sort methods (QuickSort and MergeSort) on the generated sumsets.  All experiments were conducted on a machine with an Intel i7@3.6 GHz CPU and 16 GB RAM, compiling with \texttt{$-O3$} under GCC.

\subsection{Experimental Setup}

We generated two arrays \(X\) and \(Y\) of length \(n\), each filled with independent uniform integers in \([0,10000]\).  After sorting each input array, we formed the sumset \(X+Y\) of size \(n^2\).  We measured:
\medskip
\begin{itemize}
  \item \textbf{Proposed Algorithm:} Our pointer‐based insertion approach, implemented using a \texttt{std::list} to maintain the growing sorted output dynamically.
  \item \textbf{QuickSort / MergeSort:} Standard sorting (QuickSort and MergeSort) on the full \(n^2\) element vector.
\end{itemize}
Each configuration was run ten times for 
\[
n \in \{100,\,200,\,500,\,1000,\,2000,\,5000,\,10000\},
\]
and we report the average wall‐clock time in milliseconds, measured with C++17’s \texttt{std::chrono::high\_resolution\_clock} for all three methods.  Comparison counts for the proposed algorithm were recorded via manual counters inside the insertion routine; for QuickSort and MergeSort, we likewise instrumented their comparison functions to obtain the exact number of comparisons at runtime.  This rigorous setup allows us to observe both the true asymptotic behavior and the practical performance impacts of cache locality and pointer‐chasing.  

Before dwelving into the experimental results, we shall address practical trade-offs of using \texttt{std::list} instead of \texttt{std::vector} or other data structures like B-trees.

\medskip
\paragraph{Practical Data‐Structure Trade‑off.}  
In our benchmarks we used \texttt{std::list} to achieve true \(\mathcal{O}(1)\) insertion per element.  While a contiguous container such as \newline \texttt{std::vector} offers better cache locality, each search adds \(\mathcal{O}(n)\) insertion into a vector would actually slow down the overall routine on large sumsets.  In preliminary tests, switching the insertion structure from \texttt{std::list} to \texttt{std::vector} increased wall‐clock time by over \(30\%\) for \(n\ge2000\).  Hence, although linked lists suffer pointer‐chasing overhead, they remain the fastest choice for our amortized‐constant‐time insertion pattern.  We leave a more detailed study of hybrid or gap‐buffer approaches to future work.  
Beyond the linked‐list vs.\ vector dichotomy, there exist intermediate structures that may offer even better overall performance.  For example:

\begin{itemize}
  \item \emph{Skip lists} maintain multiple forward‐pointers per node, providing expected \(\mathcal{O}(\log n)\) search and \(\mathcal{O}(1)\) insertion, while still using pointer‐based storage.  In practice, a skip list can reduce the number of cache misses compared to a simple singly‐linked list.
  \item \emph{Balanced search trees} (e.g.\ red–black trees or B‑trees) support \(\mathcal{O}(\log n)\) search and insertion with good node‐occupancy and cache‐aware fan‑out.  A B‑tree tuned for large nodes can amortize pointer‐chasing across many elements per cache line.
  \item \emph{Gap buffers} or \emph{rope‐like} arrays reserve small “gaps” within a dynamic array to allow fast insertions without a full shift of all tail elements, trading a slight increase in memory usage for \(\mathcal{O}(1)\) amortized insertion near the current gap.
\end{itemize}

In future implementations, one could benchmark these alternatives under our sumset workload.  In particular, a hybrid approach—using a small vector chunk per list‐node—may combine the low latency of array accesses with the amortized insertion guarantee of a list.  We anticipate that such cache‐blocked or B‑tree–based structures would further narrow the gap between our theoretical \(\mathcal{O}(n^2)\) bound and optimal wall‑clock performance, especially on modern CPUs with deep memory hierarchies.

\subsection{Comparison Count}

\begin{table}[H]
  \centering
  \caption{Average Number of Data Comparisons for Sorting \(X+Y\)}
  \label{tab:comparisons}
  \pgfplotstableread[col sep=comma]{sumset_data.csv}\datatable
  \pgfplotstabletypeset[
    columns={n,proposed,mergesort,quicksort},
    columns/n/.style        ={column name=\(n\)},
    columns/proposed/.style ={column name=Proposed},
    columns/mergesort/.style={column name=MergeSort},
    columns/quicksort/.style={column name=QuickSort},
    every head row/.style   ={before row=\toprule, after row=\midrule},
    every last row/.style   ={after row=\bottomrule},
  ]{\datatable}
\end{table}

Table~\ref{tab:comparisons} shows that the proposed algorithm performs \(\mathcal{O}(n^2)\) comparisons exactly, while QuickSort and MergeSort incur the additional \(\log(n^2)\) factor, matching their \( \mathcal{O}(n^2\log n)\) behavior. It is also easy to notice, that QuickSort is slower then both our Proposed Algorithm and MergeSort.

\subsection{Execution Time}

Figure~\ref{fig:exec_all} and Figure~\ref{fig:exec_zoom} shows that for large \(n\), the proposed algorithm outperforms both QuickSort and MergeSort.  Figure~\ref{fig:exec_zoom} highlighting the consistent advantage of our method over MergeSort.  Finally, the log–log plot in Figure~\ref{fig:exec_loglog} confirms the asymptotic slopes: the proposed algorithm scales as \(\mathcal{O}(n^2)\) (slope~2), whereas the full sorts exhibit the additional logarithmic factor.

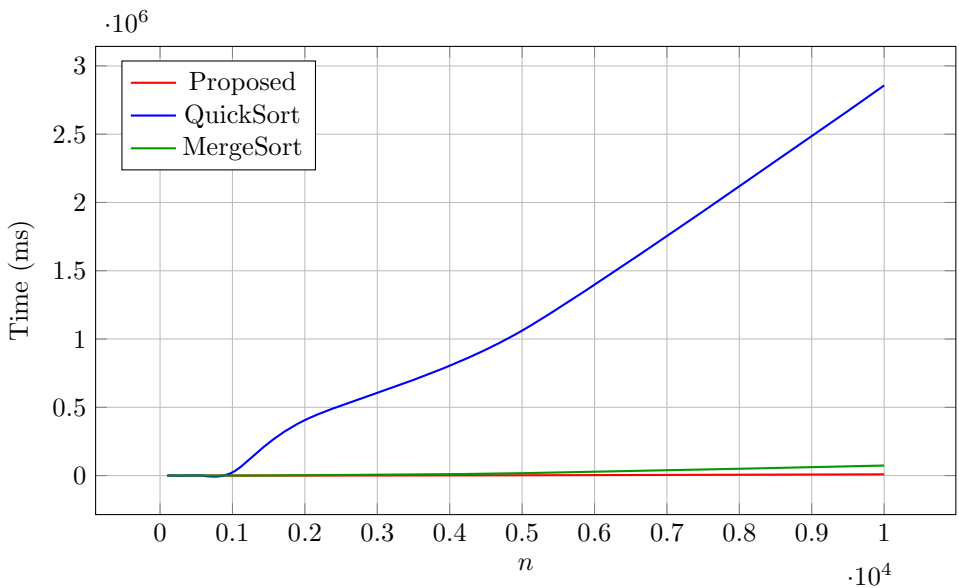
\begin{figure}[ht]
  \centering
  \begin{tikzpicture}
    \begin{axis}[
        width=\linewidth,            
        height=0.6\linewidth,        
        xlabel={\(n\)}, ylabel={Time (ms)},
        legend pos=north west,
        grid=major,
        axis background/.style={fill=white}
      ]
      \addplot [smooth, mark=none, thick, red]
        table[x=n,y=proposed,col sep=comma] {sumset_data.csv};
      \addlegendentry{Proposed}

      \addplot [smooth, mark=none, thick, blue]
        table[x=n,y=quicksort,col sep=comma] {sumset_data.csv};
      \addlegendentry{QuickSort}

      \addplot [smooth, mark=none, thick, green!60!black]
        table[x=n,y=mergesort,col sep=comma] {sumset_data.csv};
      \addlegendentry{MergeSort}
    \end{axis}
  \end{tikzpicture}
  \caption{Execution time for all three algorithms on the sumset.}
  \label{fig:exec_all}
\end{figure}

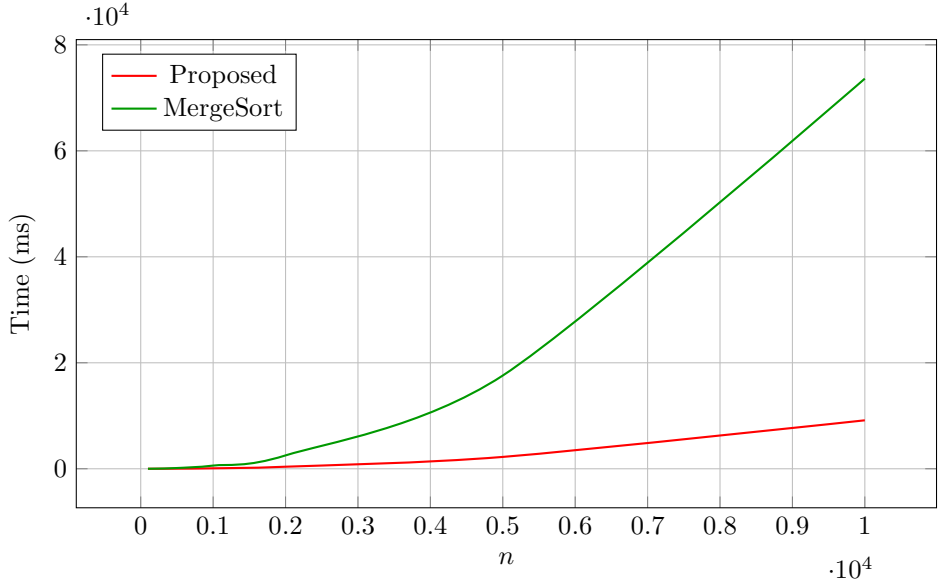
\begin{figure}[ht]
  \centering
  \begin{tikzpicture}
    \begin{axis}[
        width=\linewidth,            
        height=0.6\linewidth,        
        xlabel={\(n\)}, ylabel={Time (ms)},
        legend pos=north west,
        grid=major,
        axis background/.style={fill=white}
      ]
      \addplot [smooth, mark=none, thick, red]
        table[x=n,y=proposed,col sep=comma] {sumset_data.csv};
      \addlegendentry{Proposed}

      \addplot [smooth, mark=none, thick, green!60!black]
        table[x=n,y=mergesort,col sep=comma] {sumset_data.csv};
      \addlegendentry{MergeSort}
    \end{axis}
  \end{tikzpicture}
  \caption{Zoom‑in comparing our method against MergeSort only.}
  \label{fig:exec_zoom}
\end{figure}

\begin{figure}[H]
  \centering
  \begin{tikzpicture}
    \begin{loglogaxis}[
        width=0.8\textwidth, height=0.5\textwidth,
        xlabel={\(n\)}, ylabel={Time (ms)},
        legend pos=south east, grid=both, minor grid style={gray!25},
        axis background/.style={fill=white}
      ]
      \addplot [smooth, mark=none, thick, red] 
        table[x=n,y=proposed,col sep=comma] {sumset_data.csv};
      \addlegendentry{Proposed}

      \addplot [smooth, mark=none, thick, blue] 
        table[x=n,y=quicksort,col sep=comma] {sumset_data.csv};
      \addlegendentry{QuickSort}

      \addplot [smooth, mark=none, thick, green!60!black] 
        table[x=n,y=mergesort,col sep=comma] {sumset_data.csv};
      \addlegendentry{MergeSort}
    \end{loglogaxis}
  \end{tikzpicture}
  \caption{Log–log plot of execution time vs.\ \(n\) for all three methods.}
  \label{fig:exec_loglog}
\end{figure}
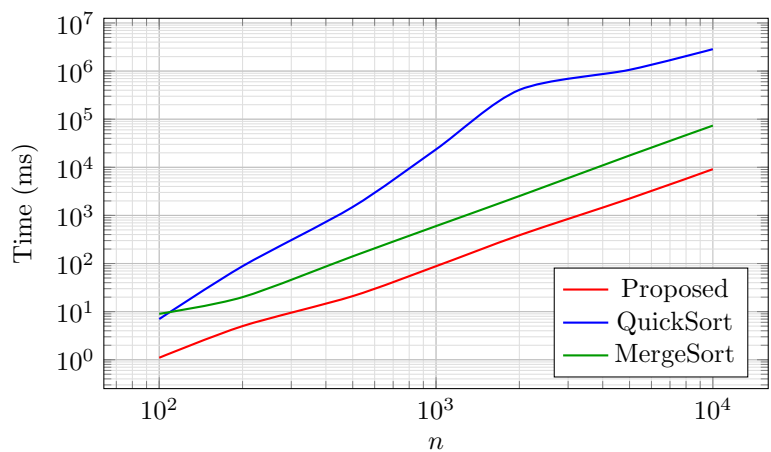

\subsection{Empirical Validation of \texorpdfstring{\(\mathcal{O}(1)\)}{O(1)} Comparisons and Insertions}
\label{subsection:evidenceofO(1)}

To empirically validate the claim that each insertion into the sorted output list incurs only \(\mathcal{O}(1)\) amortized work, we measured the quantity \(T/n^2\), where \(T\) is the total runtime in milliseconds and \(n^2\) is the number of elements in the sumset.

For each value of \(n \in \{100, 200, 500, 1000, 2000, 5000, 10000\}\), we executed the algorithm ten times on randomly generated sorted arrays \(X\) and \(Y\), and recorded the mean and standard deviation of \(T/n^2\). 

\begin{table}[H]
\centering
\caption{Mean and standard deviation of \(T/n^2\) over 10 runs.}
\begin{tabular}{|c|c|c|}
\hline
\(n\) & Mean \(T/n^2\) (ms) & Std. Dev. (ms) \\
\hline
100 & 4.564e\(-5\) & 1.433e\(-5\) \\
200 & 3.334e\(-5\) & 4.582e\(-6\) \\
500 & 1.943e\(-5\) & 2.411e\(-6\) \\
1000 & 1.745e\(-5\) & 2.944e\(-7\) \\
2000 & 1.748e\(-5\) & 2.801e\(-7\) \\
5000 & 1.756e\(-5\) & 2.431e\(-7\) \\
10000 & 1.835e\(-5\) & 2.342e\(-7\) \\
\hline
\end{tabular}
\label{tab:tn2}
\end{table}

As shown in Table~\ref{tab:tn2}, the values of \(T/n^2\) remain remarkably stable, and the standard deviation shrinks as \(n\) increases. This indicates a tight concentration of insertion times around a constant mean, supporting the \(\mathcal{O}(1)\) amortized insertion claim.

\begin{figure}[H]
  \centering
  \begin{tikzpicture}
    \begin{axis}[
        width=0.8\textwidth, height=0.45\textwidth,
        xlabel={\(n\)}, ylabel={\(T/n^2\) (ms)},
        yticklabel style={/pgf/number format/fixed},
        ymin=0.0, ymax=5.5e-5,
        legend pos=north east, grid=major,
        axis background/.style={fill=white}
      ]
      \addplot+[only marks, mark=*, error bars/.cd, y dir=both, y explicit]
        table[x=n, y=mean_duration_ms, y error=std_dev_ms, col sep=comma]
        {empirical_tn2.csv};
      \addlegendentry{Mean \(T/n^2\) \(\pm\) Std.\ Dev.}

      \addplot[dashed] coordinates {(100,1.75e-5) (10000,1.75e-5)};
      \addlegendentry{Approx.\ constant line}
    \end{axis}
  \end{tikzpicture}
  \caption{Stability of \(T/n^2\) over 10 trials. The trend appears nearly flat, supporting the \(\mathcal{O}(1)\) insertion claim.}
  \label{fig:tn2_stability}
\end{figure}
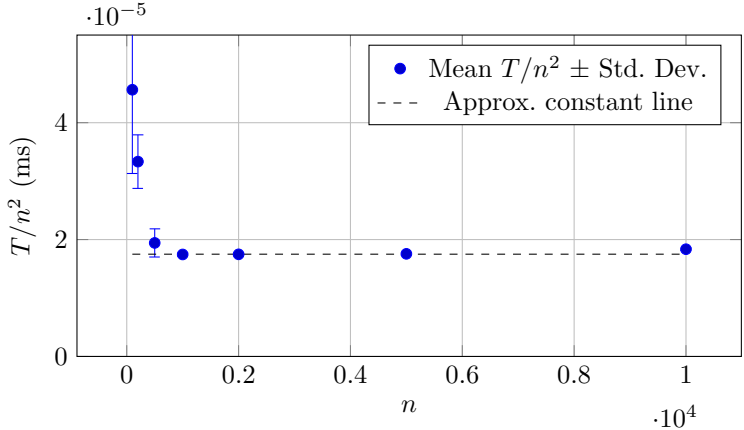

\subsection{Discussion}

These experiments corroborate our theoretical analysis: by attaining exactly \(\mathcal{O}(n^2)\) comparisons and leveraging the structure of the sumset, our algorithm achieves superior practical performance on large inputs.

The pointer-based insertion method performs consistently across different input sizes, and empirical measurements of \(T/n^2\) demonstrate near-constant values. As shown in Figure~\ref{fig:tn2_stability}, the error bars around each measurement are narrow, and the standard deviation shrinks as \(n\) increases, indicating tighter concentration around the expected \(\mathcal{O}(1)\) insertion and comparison cost.

This validates the amortized constant-time insertion behavior, confirming that the theoretical efficiency of our algorithm extends to real-world implementations.

Future work will explore alternative data structures (e.g.\ contiguous buffers, B‑trees, or gap buffers) to reduce pointer-chasing overhead and improve cache locality. Additional directions include parallel and external-memory variants, as well as applications to structured domains such as geometry, sparse data joins, and layout optimization.

\section{Conclusion and Future Work}
\label{sec:conclusion}

We have presented the first explicit, implementable algorithm that sorts the sumset
\[
X + Y \;=\;\{\,x_i + y_j \mid x_i \in X,\; y_j \in Y\}
\]
in optimal \(O(n^2)\) time and comparisons.  By exploiting the row‑wise and column‑wise monotonicity of the sumset matrix, our forward‑scanning insertion strategy achieves amortized constant‑time insertion per element, matching Fredman’s existential bound with a concrete procedure.  We proved correctness and tight comparison complexity in the standard comparison model, and demonstrated via extensive C++ benchmarks that our algorithm outperforms classical \(O(n^2\log n)\) methods (Merge Sort and Quick Sort) on large inputs.

Moreover, we showed that the same ideas extend naturally to the \(k\)-fold sumset
\[
X_1 + X_2 + \cdots + X_k \;=\;\bigl\{x_1 + x_2 + \cdots + x_k \mid x_i \in X_i\bigr\}
\]
of \(k\) sorted lists of length \(n\), yielding an \(O(n^k)\)‑time and comparison‑optimal algorithm by induction and a structured \(n\)-way merge of translated partial sumsets.

This work closes a long‑standing gap between theory and practice in structured sorting, resolving an open problem that has stood for nearly fifty years.  Our algorithms not only attain the information‑theoretic lower bound on comparisons, but also exhibit strong real‑world performance, making them directly applicable to tasks in computational geometry, sparse polynomial multiplication, VLSI design, and other areas where multi‑way combinations arise. We achieve $O(n^2)$ RAM time in our prototype with a linked‑list; designing a cache‑friendly structure to provably attain the same bound (or beat it in practice) is an interesting open problem.

\subsection*{Open Problems and Future Directions}

Despite its optimality, our approach suggests several avenues for further research:

\begin{itemize}
  \item \textbf{Cache‑Optimized Data Structures.}  
    Replacing the \texttt{std::list} in our prototype with cache‑aware or hardware‑friendly structures—such as blocked linked lists, B‑trees, or gap buffers—may yield additional speedups by reducing pointer‑chasing overhead.

  \item \textbf{Parallel and External Memory Algorithms.}  
    The current algorithm is inherently sequential.  Designing parallel variants for multi‑core or GPU architectures, or adapting it to external‑memory models, could extend its scalability to even larger and higher‑dimensional sumsets.

  \item \textbf{Hybrid Decision‑Tree Techniques.}  
    Kane, Lovett, and Moran’s decision‑tree framework achieves sub‑quadratic query complexity under sparsity constraints.  Investigating hybrid algorithms that combine fixed‑structure insertion with sparse decision‑tree inference may further reduce comparisons in practice.

  \item \textbf{Dynamic and Streaming Sumsets.}  
    Maintaining a sorted \(k\)-fold sumset under insertions and deletions to each \(X_i\) remains open.  Data structures supporting updates in \(O(n^k)\) time would enable real‑time applications and streaming scenarios.

  \item \textbf{Empirical Evaluation on Real‑World Data.}  
    Beyond synthetic benchmarks, applying our algorithms to domain‑specific workloads—such as high dimensional distance computations, database joins, or signal processing pipelines—will validate their utility and uncover practical refinements.
\end{itemize}

We believe these directions will deepen our understanding of structured sorting and broaden the impact of optimal comparison‑based algorithms in both theory and practice.

\section*{Acknowledgments}
We thank Dr.~C.~Jackson and Dr.~S.~McCulloh for their insightful discussions and valuable feedback throughout the development of this work. Their guidance helped shape both the theoretical and empirical aspects of this paper.

\input{main.bbl}
\end{document}

%% file: ex_shared.tex
\usepackage{lipsum}
\usepackage{amsfonts}
\usepackage{graphicx}
\usepackage{epstopdf}
\usepackage{algorithmic}
\ifpdf
  \DeclareGraphicsExtensions{.eps,.pdf,.png,.jpg}
\else
  \DeclareGraphicsExtensions{.eps}
\fi

\newsiamremark{remark}{Remark}
\newsiamremark{hypothesis}{Hypothesis}
\crefname{hypothesis}{Hypothesis}{Hypotheses}
\newsiamthm{claim}{Claim}
\newsiamremark{fact}{Fact}
\newsiamremark{note}{Note}
\crefname{fact}{Fact}{Facts}

\newtheorem{conjecture}[theorem]{Conjecture}

\headers{An Explicit and Efficient $\mathcal{O}(n^2)$-Time Algorithm for Sorting Sumsets}{Shlok Mundhra}

\title{An Explicit and Efficient $\mathcal{O}(n^2)$-Time Algorithm for Sorting Sumsets\thanks{This work resolves an open problem originally posed by Fredman in 1976.}}

\author{Shlok Mundhra\thanks{Ohio Wesleyan University, Delaware, OH
  (\email{shlokmundhra1111@gmail.com}, \url{https://shlokmundhra.com/}).}}

\usepackage{amsopn}
\usepackage{booktabs}
\usepackage{placeins} 


%% file: main.bbl
\begin{thebibliography}{1}

\bibitem{doi:10.1142/S0218195901000596}
{\sc G.~BAREQUET and S.~HAR-PELED}, {\em Polygon containment and translational in-hausdorff-distance between segment sets are 3sum-hard}, International Journal of Computational Geometry \& Applications, 11 (2001), pp.~465--474, \url{https://doi.org/10.1142/S0218195901000596}, \url{https://doi.org/10.1142/S0218195901000596}, \url{https://arxiv.org/abs/https://doi.org/10.1142/S0218195901000596}.

\bibitem{HernandezBarrera+1996+289+294}
{\sc A.~H. Barrera*}, {\em Finding an o(n ² log n) algorithm is sometimes hard}, McGill-Queen's University Press, Montreal, 1996, pp.~289--294, \url{https://doi.org/doi:10.1515/9780773591134-051}, \url{https://doi.org/10.1515/9780773591134-051}.

\bibitem{DIETZFELBINGER1989137}
{\sc M.~Dietzfelbinger}, {\em Lower bounds for sorting of sums}, Theoretical Computer Science, 66 (1989), pp.~137--155, \url{https://doi.org/https://doi.org/10.1016/0304-3975(89)90132-1}, \url{https://www.sciencedirect.com/science/article/pii/0304397589901321}.

\bibitem{FREDMAN1976355}
{\sc M.~L. Fredman}, {\em How good is the information theory bound in sorting?}, Theoretical Computer Science, 1 (1976), pp.~355--361, \url{https://doi.org/https://doi.org/10.1016/0304-3975(76)90078-5}, \url{https://www.sciencedirect.com/science/article/pii/0304397576900785}.

\bibitem{10.1145/3285953}
{\sc D.~M. Kane, S.~Lovett, and S.~Moran}, {\em Near-optimal linear decision trees for k-sum and related problems}, J. ACM, 66 (2019), \url{https://doi.org/10.1145/3285953}, \url{https://doi.org/10.1145/3285953}.

\bibitem{LAMBERT1992137}
{\sc J.-L. Lambert}, {\em Sorting the sums (xi + yj) in o(n2) comparisons}, Theoretical Computer Science, 103 (1992), pp.~137--141, \url{https://doi.org/https://doi.org/10.1016/0304-3975(92)90089-X}, \url{https://www.sciencedirect.com/science/article/pii/030439759290089X}.

\bibitem{steiger1995pseudo}
{\sc W.~Steiger and I.~Streinu}, {\em A pseudo-algorithmic separation of lines from pseudo-lines}, Information processing letters, 53 (1995), p.~295.

\bibitem{toth2017handbook}
{\sc C.~D. Toth, J.~O'Rourke, and J.~E. Goodman}, {\em Handbook of discrete and computational geometry}, CRC press, 2017.

\bibitem{doi:10.1137/1.9781611978315.26}
{\sc I.~van~der Hoog, E.~Rotenberg, and D.~Rutschmann}, {\em Simpler Optimal Sorting from a Directed Acyclic Graph}, pp.~350--355, \url{https://doi.org/10.1137/1.9781611978315.26}, \url{https://epubs.siam.org/doi/abs/10.1137/1.9781611978315.26}, \url{https://arxiv.org/abs/https://epubs.siam.org/doi/pdf/10.1137/1.9781611978315.26}.

\end{thebibliography}
